\documentclass[envcountsame]{llncs}
\usepackage{breakcites}
\usepackage{llncsdoc}
\usepackage{booktabs} 
\usepackage{bbm}
\usepackage{xspace}
\usepackage{amsmath}
\usepackage{amsfonts}
\usepackage{amssymb}

\usepackage{thmtools}
\usepackage{thm-restate}

\newcommand{\mathsc}[1]{{\normalfont\textsc{#1}}}
\renewcommand{\hat}{\widehat}
\renewcommand{\tilde}{\widetilde}
\newcommand{\norm}[1]{\lVert#1{\rVert}}
\newcommand{\normone}[1]{{\norm{#1}}_1}
\newcommand{\norminf}[1]{{\norm{#1}}_\infty}
\newcommand{\R}{{\mathbb{R}}}
\newcommand{\mone}{{\mathbbm{1}}}
\newcommand{\indic}[1]{{\mone\!\left[#1\right]}}
\providecommand{\Prob}[2]{\ensuremath{{\Pr_{#1}\!\left[#2\right]}}}
\providecommand{\expect}[2]{\ensuremath{{\mathbb{E}_{#1}\!\left[#2\right]}}}
\newcommand{\var}{\operatorname{Var}}
\newcommand{\cov}{\operatorname{Cov}}
\newcommand{\supp}{{\mathrm{supp}}}

\newcommand{\maxi}[1]{\mbox{maximize} & {#1 } & \\}

\newcommand{\st}{\mbox{subject to} }
\newcommand{\con}[1]{&#1 & \\}
\newcommand{\qcon}[2]{&#1, &#2.  \\}
\newenvironment{lp}{\begin{equation}  \begin{array}{lll}}{\end{array}\end{equation}}
\newenvironment{lp*}{\begin{equation*}  \begin{array}{lll}}{\end{array}\end{equation*}}

\usepackage{color}

\definecolor{WildStrawberry}{RGB}{255,67,164}

\newcommand{\budget}{{b}}
\newcommand{\capsum}{{c}}
\newcommand{\vcap}{{s}}
\newcommand{\Rev}{{\mathsc{Rev}}}
\newcommand{\RevB}{{\mathsc{Rev}^\budget}}

\newcommand{\SRev}{{\mathsc{SRev}}}
\newcommand{\SRevB}{{\mathsc{SRev}^\budget}}
\newcommand{\BRev}{{\mathsc{BRev}}}
\newcommand{\BRevB}{{\mathsc{BRev}^\budget}}

\newcommand{\SINGLE}{{\mathsc{Single}}}
\newcommand{\CORE}{{\mathsc{Core}}}
\newcommand{\TAIL}{{\mathsc{Tail}}}
\newcommand{\VVAL}{{\Phi}} 
\newcommand{\vVAL}{{\varphi}} 
\newcommand{\vVALt}{{\tilde \varphi}} 

\newcommand{\thmskip}{\vspace{8pt plus 2pt minus 4pt}}

\begin{document}
\title{A Simple Mechanism \\ for a Budget-Constrained Buyer}
\author{Yu Cheng\inst{1}\and
Nick Gravin\inst{2} \and
Kamesh Munagala\inst{1} \and
Kangning Wang\inst{1}}
%

%
\institute{Duke University;
\email{\{yucheng,kamesh,knwang\}@cs.duke.edu}\and
Shanghai University of Finance and Economics;
\email{nikolai@mail.shufe.edu.cn}}

\maketitle



\begin{abstract}
We study a classic Bayesian mechanism design setting of monopoly problem for an additive buyer in the presence
of budgets. In this setting a monopolist seller with $m$ heterogeneous items faces a single buyer and seeks 
to maximize her revenue. The buyer has a budget and additive valuations drawn independently for each item from (non-identical) distributions.
We show that when the buyer's budget is publicly known, the better of selling each item separately and selling the grand bundle extracts a constant fraction of the optimal revenue.
When the budget is private, we consider a standard Bayesian setting where buyer's budget $b$ is drawn from a known distribution $B$. We show that if $b$ is independent of the valuations and distribution $B$
satisfies monotone hazard rate condition, then selling items separately or in a grand bundle 
is still approximately optimal. We give a complementary example showing that no constant approximation simple mechanism is possible if budget $b$ can be interdependent with valuations.
%
%
\end{abstract}


\section{Introduction}
Revenue maximization is one of the fundamental problems in auction theory.
The well-celebrated result of Myerson~\cite{Myerson81} characterized the revenue-maximizing mechanism when there is only one item for sale.
Specifically, in the single buyer case, the optimal solution is to post a take-it-or-leave-it price.
Since Myerson's work, the optimal mechanism design problem has been studied extensively in computer science literature and much progress has been made~\cite{CaiDW12a,CaiDW12b,CaiDW13a,CaiDW13b,AlaeiFHHM12,Daskalakis15}.
The problem of finding the optimal auction turned out to be so much more complex than the single-item case. 
Unlike the Myerson's single-item auction, the optimum can use randomized allocations and price bundles of items already for two items and a single buyer. It is also known that the gap between the revenue of the optimal randomized and optimal deterministic mechanism can be arbitrarily large~\cite{BriestCKW10,HartN13}, the optimal mechanism may require a menu with infinitely many options~\cite{manelli2007multidimensional,DaskalakisDT13}, and the revenue of the optimal auction may decrease when the buyer's valuation distributions move upwards (in the stochastic dominance sense).

In light of these negative results for optimal auction design, many recent papers focused on the design of \emph{simple} mechanisms that are \emph{approximately} optimal.
One such notable line of work initiated by Hart and Nisan~\cite{HartN17}
concerns a basic and natural setting of monopoly problem for the 
buyer with item values drawn independently from given distributions $D_1,\ldots,D_m$ and
whose valuation for the sets of items is additive\footnote{A buyer has additive valuations if his value for a set of items is equal to the sum of his values for the items in the set.} (linear). 
A remarkable result by Babaioff~et~al.~\cite{BabaioffILW14} showed 
that the better mechanism of either selling items separately, or selling the grand bundle extracts at least $(1/6)$-fraction of the optimal revenue.
It was also observed~\cite{HartN13,BabaioffILW14,RubinsteinW15} that the independence assumption on the items is essentially necessary and without it no simple (any deterministic) mechanism cannot be approximately optimal.


Auction design with budget constraints is an even harder problem.
Because buyer's utility is no longer quasi-linear, many standard concepts do not carry over\footnote{E.g., the classic VCG mechanism may not be implementable and social efficiency may not be achievable in the budgeted-setting~\cite{Singer10}.}. 
For example,
even for one buyer and one item, the optimal mechanism may require randomization when the budget is public~\cite{ChawlaMM11}, and may need an exponential-size menu when the budget is private~\cite{DevanurW17}.
Despite many efforts~\cite{LaffontR96,CheG00,GoldbergHW01,BorgsCIMS05,Abrams06,ChenGL11,DobzinskiLN12,BhattacharyaGGM10,BhattacharyaCMX10,ChakrabartyG10,ChawlaMM11,BeiCGL12,GoelML15,DevanurHH13,DaskalakisDW15,DevanurW17,Singer10}, the theory of optimal auction design with budgets is still far behind the theory without budgets.

In this paper, we investigate the effectiveness of simple mechanisms in the presence of budgets.
Our work is motivated by the following questions:

\begin{quote}
\em How powerful are simple mechanisms in the presence of budgets? In particular, is there a simple mechanism that is approximately optimal for a budget-constrained buyer with independent valuations?
\end{quote}

To this end we consider one of the most basic and natural 
settings of extensively studied monopoly problem for an additive buyer. 
In this setting, a monopolistic seller sells $m$ items to a single buyer.
The buyer has additive valuations drawn independently for each item from an arbitrary (non-identical) distribution.
We study two different budget settings: the \emph{public budget} case where the buyer has a fixed budget known to the seller, and the \emph{private budget} case where the buyer's budget is drawn from a distribution.
The seller wishes to maximize her revenue by designing an auction subject to individual rationality, incentive compatibility, and budget constraints.
We consider the Bayesian setting where the buyer knows his budget and his values for each item, but the seller only knows the prior distributions.

\subsection{Our Results and Techniques}
Our first result is that simple mechanisms remain approximately optimal when the buyer has a public budget.
\begin{theorem}
\label{thm:main-informal}
For an additive buyer with a known public budget and independent valuations, the better of selling each item separately and selling the grand bundle extracts a constant fraction of the optimal revenue.
\end{theorem}
Theorem~\ref{thm:main-informal} is among the few positive results in budget-constrained settings that hold for arbitrary distributions.
Before our work, it is not clear that any mechanism extracting a constant fraction of the optimal revenue can be computed in polynomial time.

In Sections~\ref{sec:proof1}~and~\ref{sec:overview},
  we present two different approaches to prove Theorem~\ref{thm:main-informal}.
Both approaches truncate the valuation distribution $V$ according to the budget $b$ (in different ways) and then relate the revenues of the optimal/simple mechanisms on the truncated distribution to the revenues on the original valuations.
The first approach uses the main result of~\cite{BabaioffILW14} in a black-box way, and the second approach adapts the duality-based framework developed in~\cite{CaiDW16}.

It is worth pointing out that many of our structural lemmas hold for correlated valuations as well.
Using these lemmas, we can generalize Theorem~\ref{thm:main-informal} with minimum effort to allow the buyer to have weakly correlated valuations.
We call a distribution $\hat V$ \emph{weakly correlated} if it is the result of conditioning an independent distribution $V$ on the sum of $v \sim V$ being at most $\capsum$: $\hat V = V_{|(\sum v_i \le c)}$ (See Definition~\ref{def:weakly} for the formal definition).
\begin{corollary}
\label{cor:weakly}
Let $\hat V$ be a weakly correlated distribution. 
For an additive buyer with a public budget and valuations drawn from $\hat V$, the better of selling separately and selling the grand bundle extracts a constant fraction of the optimal revenue.
\end{corollary}

In Section~\ref{sec:private}, we examine the private budget setting.
The budget $b$ is no longer fixed but is drawn from a distribution $B$.
The seller only knows the prior distribution $B$ but not the value of $b$.
We first show that if the valuations can be correlated with the budget, the problem is at least as hard as budget-free mechanism design with correlated valuations, where simple mechanisms are known to be ineffective.
In light of this negative result, we focus on the setting where the budget distribution $B$ is independent of the valuations $V$.
In this setting, we show that simple mechanisms are approximately optimal when the budget distribution satisfies the monotone hazard rate (MHR) condition.
\begin{theorem}
\label{thm:private-mhr}
When the budget distribution $B$ is MHR, the better mechanism of pricing items separately and selling a grand bundle  achieves a constant fraction of the optimal revenue.
\end{theorem}
We will show that it is sufficient to pretend the buyer has a public budget $b^* = \expect{b \sim B}{b}$.
The proof of Theorem~\ref{thm:private-mhr} uses the MHR condition, as well as the fact that for a public budget $b$, the (budget-constrained) optimal revenue is nondecreasing in $b$, but optimal revenue divided by $b$ is nonincreasing in $b$.

\subsection{Related Work}
The most closely related to ours are the following two lines of work.
\paragraph{\bf Simple Mechanisms.}
In a line of work initiated by Hart and Nisan~\cite{HartN17,LiY13,BabaioffILW14}, \cite{BabaioffILW14} first showed that for an additive buyer with independent valuations, either selling separately or selling the grand bundle extracts a constant fraction of the optimal revenue.
This was later extended to multiple buyers~\cite{Yao15}, as well as buyers with more general valuations (e.g., sub-additive~\cite{RubinsteinW15}, valuations with a common-value component~\cite{BateniDHS15}, and valuations with complements~\cite{EdenFFTW17}).
Others have studied the trade-off between the complexity and approximation ratio of an auction, along with the design of small-menu mechanisms in various settings~\cite{HartN13,TangW17,DughmiHN14,ChengCDEHT15,BabaioffGN17}.



\paragraph{\bf Auctions for Budget-Constrained Buyers.}
There has been a lot of work studying the impact of budget constraints on mechanism design.
Most of the earlier work required additional assumptions on the valuations distributions, like regularity or monotone hazard rate~(\cite{LaffontR96,CheG00,BhattacharyaGGM10,PaiV14}).
We mention a few results that work for arbitrary distributions.
For public budgets, \cite{ChawlaMM11} designed approximately optimal mechanisms for several single-parameter settings and multi-parameter settings with unit-demand buyers.
For private budgets, \cite{DevanurW17} characterized the structure of the optimal mechanism for one item and one buyer.
\cite{DaskalakisDW15} gave a constant-factor approximation for additive bidders whose private budgets can be correlated with their values. However, they require the buyers' valuation distribution to be given explicitly, which is of exponential size in our setting.
There are also approximation and hardness results in the prior-free setting~\cite{BorgsCIMS05,Abrams06,DevanurHH13}, as well as designing Pareto optimal auctions~\cite{DobzinskiLN12,GoelML15}.

\paragraph{\bf Other Related Work.} Our work concerns revenue maximization for additive buyer. 
Another natural and basic scenario extensively studied in the literature concerns buyers with 
unit-demand preferences~\cite{ChawlaHK07,ChawlaHMS10,ChawlaMS15}.
Our work studies monopoly problem for additive budgeted buyer in the standard Bayesian approach. In this framework, the prior distribution is known to the seller and typically is assumed to be independent. Parallel to this framework, the (budgeted) additive monopoly problem has been studied in a new robust optimization framework~\cite{carroll2017robustness,GravinL18}. Another group of papers on budget feasible mechanism design~\cite{BeiCGL12,ChenGL11,SinglaK13,Singer10} studies different reverse auction settings and are concerned with value maximization.


\section{Preliminaries}
\subsection{Optimal Mechanism Design}
We study the design of optimal auctions with one buyer, one seller, and $m$ heterogeneous items labeled by $[m] = \{1, \ldots, m\}$.
There is exactly one copy of each item, and the items are indivisible.
The buyer has additive valuation ($v(S)=\sum_{j\in S}v(\{i\})$ for any set $S\subseteq [m]$) and a publicly known budget $b$~\footnote{
In this paper, we mostly focus on the public budget case. So we define notations and discuss backgrounds assuming the buyer has a public budget.}.

We use $v \in \R^m$ to denote the buyer's valuations, where $v_j$ is the buyer's value for item $j$.
We consider the Bayesian setting of the problem, in which the buyer's values are drawn from a discrete\footnote{Like previous work on simple and approximately optimal mechanisms, our results extend to continuous types as well (see, e.g.,~\cite{CaiDW16} for a more detailed discussion).} distribution $V$. 
Let $T = \supp(V)$ be the set of all possible valuation profiles in $V$.
We use $f(t)$ for any $t\in T$ to denote the probability mass function of $V$: $f(t) = \Prob{v \sim V}{v = t}$.
Let $T_j = \supp(V_j)$.
We say the valuation distribution $V$ is independent across items if it can be expressed as $V = \times_j V_{j}$.

We assume the buyer is risk-neutral and has quasi-linear utility when the payment does not exceed his budget.
Let $\pi: T \rightarrow [0,1]^m$ and $p:T \rightarrow \R$ denote the allocation and payment rules of a mechanism respectively.
That is, when the buyer reports type $t$, the probability that he will receive item $j$ is $\pi_j(t)$, and his expected payment is $p(t)$ (over the randomness of the mechanism).
Thus, if the buyer has type $t$, his (expected) value for reporting type $t'$ is exactly $\pi(t')^\top t$,~\footnote{We use $x^\top y = \sum_{i=1}^m x_i y_i$ to denote the inner product of two vectors $x$ and $y$.} and
  his (expected) utility for reporting type $t'$ is
\[
u(t, t') = \begin{cases}
\pi(t')^\top t - p(t') & \text{if } p(t') \le \budget, \text{ and } \\
-\infty & \text{otherwise.}
\end{cases}
\]

By the revelation principle, it is sufficient to consider mechanisms that are incentive compatible (i.e., ``truthful'').
A mechanism $M = (\pi, p)$ is (interim) incentive-compatible (IC) if the buyer is incentivized to tell the truth (over the randomness of mechanism), and (interim) individually rational (IR) if the buyer's expected utility is non-negative whenever he reports truthfully.
We use $\varnothing$ for the option of not participating in the auction ($\pi(\varnothing) = 0, p(\varnothing) = 0$), and let $T^+ = T \cup \{\varnothing\}$.
Then, the IC and IR constraints can be unified as follows:
\[
u(t, t) \ge u(t, t') \quad \forall t \in T, t' \in T^+.
\]
To summarize, when the seller faces a single buyer with budget $b$ and valuation drawn from $V$, the optimal mechanism $M^* = (\pi^*, p^*)$ is the optimal solution to the following (exponential-size) linear program (LP):
\begin{lp}
\label{lp:budget-reduceform}
\maxi{\sum_{t \in T} f(t) p(t)}
\st \qcon{\pi(t')^\top t - p(t') \le \pi(t)^\top t - p(t)}{\forall t \in T, t' \in T^+}
    \qcon{0 \le \pi_j(t) \le 1}{\forall t \in T, j \in [m]}
    \qcon{p(t) \le \budget}{\forall t \in T}
    \con{\pi(\varnothing) = 0, \; p(\varnothing) = 0.}
\end{lp}

A mechanism is called ex-post IC, ex-post IR, or ex-post budget-preserving respectively, if the corresponding constraints hold for all possible outcomes, without averaging over the randomness in the mechanism.
We will show the better of pricing each item separately and pricing the grand bundle, which is a deterministic ex-post mechanism, can extract a constant fraction of the revenue of any interim mechanism.

\subsection{Simple Mechanisms}

For a buyer with valuation distribution $V$, we frequently use the following notations in our analysis:
\begin{itemize}
\item $\Rev(V)$: the revenue of the optimal truthful mechanism.
\item $\SRev(V)$: the maximum revenue obtainable by pricing each item separately.
\item $\BRev(V)$: the maximum revenue obtainable by pricing the grand bundle.
\item $\RevB(V)$: the revenue of the optimal truthful mechanism, when the buyer has a budget $\budget$.
\item $\SRevB(V)$: the maximum revenue that can be extracted by pricing each item separately, when the buyer has a public budget $\budget$.
\item $\BRevB(V)$: the maximum revenue that can be extracted by pricing the grand bundle, when the buyer has a public budget $\budget$.
\end{itemize}
We know that $\SRev(V)$ is obtained by running Myerson's optimal auction separately for each item, 
  and $\BRev(V)$ is obtained by running Myerson's auction viewing the grand bundle as one item.
Similarly, $\BRevB(V)$ is a single-parameter problem as well, with the minor change that the posted price is at most $\budget$.

The case of $\SRevB(V)$ is more complicated. For example, when a budgeted buyer of type $t \in \R^m$ participates in an auction with posted price $p_j$ for each item $j$, he will maximize his utility by solving a $\mathsc{Knapsack}$ problem.
There exists a poly-time computable mechanism that extracts a constant fraction of $\SRevB(V)$ (e.g.,~\cite{BhattacharyaCMX10}).
We focus on the structural result that the better of $\SRevB(V)$ and $\BRevB(V)$ is a constant approximation of $\RevB(V)$.
A better approximation for $\SRevB(V)$ is an interesting open problem that is beyond the scope of this paper.

\subsection{Weakly Correlated Distributions}
We call a distribution like $\hat V$ \emph{weakly correlated} if the only condition causing the correlation is a cap on its sum.
\begin{definition}
\label{def:weakly}
\
For an $m$-dimensional independent distribution $V$ and a threshold $\capsum > 0$, we remove the probability mass on any $t \in \supp(V)$ with $\normone{t} > \capsum$ and renormalize.
Let $\hat V := V_{|(\normone{v} \le \capsum)}$ denote the resulting distribution. Formally,
\[
\Prob{\hat v \sim \hat V}{\hat v = t} = \Prob{v \sim V}{v = t \mid \normone{v} \le \capsum}, \; \forall t \in \supp(V).
\]
\end{definition}
Weakly correlated distributions arise naturally in our analysis.
We will show that if the buyer's valuations are weakly correlated, then the better of selling separately and selling the grand bundle is approximately optimal, and this holds with or without a (public) budget constraint.

\subsection{First-Order Stochastic Dominance}
\label{sec:prelim-preceq}
Stochastic dominance is a partial order between random variables. A random variable $X$ with $\supp(X) \subseteq \R$ \emph{(weakly) first-order stochastically dominates} another random variable $Y$ with $\supp(Y) \subseteq \R$ if and only if
\[ \Prob{}{X \ge a} \ge \Prob{}{Y \ge a} \text{ for all } a \in \R. \]


This notion of stochastic dominance can be extended to multi-dimensional distributions.
In this paper, we use the notion of coordinate-wise dominance.
\begin{definition}
Given two $m$-dimensional distributions $D_1$ and $D_2$, we say $D_1$ \emph{coordinate-wise stochastic dominates} $D_2$ ($D_1 \succeq D_2$ or $D_2 \preceq D_1$) if there exists a randomized mapping $f: \supp(D_1) \rightarrow \supp(D_2)$ such that $f(x) \sim D_2$ when $x \sim D_1$, and $f(x) \leq x$ coordinate-wise for all $x \in \supp(D_1)$ with probability $1$.
\end{definition}
This notion helps us express the monotonicity of optimal revenues in some cases.
For example, we can show that $\SRev(V_1) \geq \SRev(V_2)$ when $V_1 \succeq V_2$.
The mapping $f$ allows us to couple the draws $v_1 \sim V_1$ and $v_2 \sim V_2$, so that for a set of fixed prices, if the buyer buys an item under $v_2$, he will also buy it under $v_1$.


\section{Public Budget}
\label{sec:proof1}

In this section, we focus on the public budget case and prove our main result (Theorem~\ref{thm:main-informal}).
The buyer has a fixed budget $b$ and valuations drawn from an independent distribution $V$.

\thmskip
{\noindent \bf Theorem~\ref{thm:main-informal}.~}
{\em
$\RevB(V) \le 8 \SRevB(V) + 24 \BRevB(V)$.
}
\thmskip

It follows that the better of $\SRevB(V)$ and $\BRevB(V)$ is at least $\frac{\RevB(V)}{32}$.~\footnote{We do not optimize the constants in our proofs. In Section~\ref{sec:overview}, we will give an alternative proof of Theorem~\ref{thm:main-informal} that shows $\RevB(V) \le 5 \SRevB(V) + 6 \BRevB(V)$, thus improving this constant from 32 to 11. 
}

\paragraph{\bf Overview of Our Approach.}
Instead of taking the Lagrangian dual of LP~\eqref{lp:budget-reduceform} to derive an upper bound on the optimal objective value $\RevB(V)$, we adopt a more combinatorial approach.
Intuitively, we come up with a charging argument that splits $\RevB(V)$ and charges each part to either $\SRevB(V)$ or $\BRevB(V)$.

First, we partition the buyer types $t \in \supp(V)$ into two sets: \emph{high-value} types where $\norminf{t} \ge \budget$ and \emph{low-value} types where $\norminf{t} < \budget$.
Note that we can already charge the revenue of high-value types to $\BRevB(V)$: If we sell the grand bundle at price $\budget$, all high-value types will exhaust their budgets.

We now examine the low-value types.
Let $V'$ denote the valuation distribution conditioned on the buyer having a low-value type.
Observe that $V'$ is independent because it is defined using $\ell_\infty$-norm, and we can remove the budget to upper bound its revenue.
For a budget-free additive buyer with independent valuations, we can apply the main result of~\cite{BabaioffILW14}, which states that either selling separately or grand bundling works for $V'$: $\Rev(V') = O(\SRev(V') + \BRev(V'))$.

Next, we will relate $\SRev(V'), \BRev(V')$ to $\SRevB(V'), \BRevB(V')$.
We can assume the sum of $v' \sim V'$ is usually much smaller than $\budget$.
Similar to standard tail bounds, if the sum $\normone{v'}$ is often small and the random variables are independent and bounded (each $v'_j$ is at most $\budget$), then $\normone{v'}$ must have an exponentially decaying tail.
Therefore, we can add back the budget, because the sum $\normone{v'}$, which upper bounds the buyer's payment, is rarely very large.

Finally, we will show that $\SRevB(V') = O(\SRevB(V))$ and $\BRevB(V') \le \BRevB(V)$.
The $\BRev$ statement is easy to verify, but the $\SRev$ statement is more tricky.
The monotonicity of $\SRev(V)$ in the budget-free case (see Section~\ref{sec:prelim-preceq}) no longer holds when there is a budget.
Fortunately, we can pay a factor of two and circumvent this non-monotonicity due to budget constraints.

We will now make our intuitions formal and present three key lemmas. 
Throughout the paper, we will always use $V' = V_{|\norminf{v}\le b}$ as defined below.

\begin{definition}
Fix an $m$-dimensional distribution $V = \times V_j$.
Let $V'$ be the independent distribution where every coordinate of $V$ is capped at $\budget$.
That is, $V' = \times_j V'_j$, and $V'_j$ is given by $\Prob{V'_j}{x} = \Prob{v_j \sim V_j}{\min(v_j, \budget) = x}$.
\end{definition}

\begin{lemma}
\label{lem:opt-revvp}
$\RevB(V) \le \Rev(V') + \BRevB(V)$.
\end{lemma}

\begin{lemma}
\label{lem:revvp-revbvp}
Assume $\BRev^\budget(V') < \frac{\budget}{10}$.
Then, $\BRev(V') \le 3 \BRevB(V')$ and $\SRev(V') \le \SRev^{\budget}(V') + 4\BRev^{\budget}(V')$.
\end{lemma}

\begin{lemma}
\label{lem:vp-sbrev}
$\BRevB(V') \le \BRevB(V)$ and $\SRevB(V') \le 2\SRevB(V)$.
\end{lemma}

We defer the proofs of these lemmas to Sections~\ref{sec:opt-revvp},~\ref{sec:revvp-revbvp},~and~\ref{sec:vp-sbrev}, and first use them to prove Theorem~\ref{thm:main-informal}.

\begin{proof}[of Theorem~\ref{thm:main-informal}]
If $\BRevB(V') \ge \frac{b}{10}$, then the theorem holds because the optimal revenue $\RevB(V)$ is at most the budget $\budget$.
By Lemma~\ref{lem:vp-sbrev}, $\BRevB(V) \ge \BRevB(V') \ge \frac{b}{10} \ge \frac{\RevB(V)}{10}$.

We now assume $\BRevB(V') < \frac{b}{10}$. 
The theorem follows straightforwardly from Lemmas~\ref{lem:opt-revvp},~\ref{lem:revvp-revbvp},~\ref{lem:vp-sbrev}, and a black-box use of the main result of~\cite{BabaioffILW14}.
\begin{align*}
\RevB(V) & \le \Rev(V') + \BRevB(V) \tag{Lemma~\ref{lem:opt-revvp}} \\
  & \le 4\SRev(V') + 2\BRev(V') + 2 \BRevB(V) \tag{\cite{BabaioffILW14}} \\
  & \le 4\SRevB(V') + 22 \BRev^{\budget}(V') + 2 \BRevB(V). \tag{Lemma~\ref{lem:revvp-revbvp}} \\
  & \le 8\SRevB(V) + 24 \BRevB(V). \tag*{(Lemma~\ref{lem:vp-sbrev}) \qed}
\end{align*}
\end{proof}

\subsection{Proof of Lemma~\ref{lem:opt-revvp}}
\label{sec:opt-revvp}
We will prove the following lemma, which is a generalization of Lemma~\ref{lem:opt-revvp}.

\begin{lemma}
\label{lem:revbvsplit}
Fix $b > 0$ and $0 < \capsum \le \budget$.
For any distribution $\hat V$ with $\supp(\hat V) \subseteq \supp(V)$ and $\Prob{\hat V}{t} \ge \Prob{V}{t}$ for any $\normone{t} \le \capsum$,
  we have $\RevB(V) \le (\budget / \capsum) \cdot \BRevB(V) + \Rev(\hat V)$.
\end{lemma}
  
Lemma~\ref{lem:opt-revvp} follows immediately from Lemma~\ref{lem:revbvsplit} by choosing $c = \budget$ and $\hat V = V'$, because capping each coordinate at $c$ does not create new support, and does not decrease probability mass on any type $t$ whose sum is at most $c$.

Intuitively, Lemma~\ref{lem:revbvsplit} upper bounds the optimal revenue by splitting the buyer types $t$ into two sets:
  when $\normone{t} > c$, we upper bound the seller's revenue by the budget $b$;
  when $\normone{t} \le c$, we run the optimal mechanism for $\RevB(V)$.

\begin{proof}[of Lemma~\ref{lem:revbvsplit}]
Let $T$ and $\hat T$, $f$ and $\hat f$ denote the support and probability density function of $V$ and $\hat V$ respectively.
Let $M^* = (\pi^*, p^*)$ be the optimal mechanism that obtains $\RevB(V)$.
Recall that $\pi^*$ and $p^*$ are the allocation and payment rules, and $(\pi^*, p^*)$ is the optimal solution to LP~\eqref{lp:budget-reduceform} for $f$ and $T$.

We split the optimal revenue into two parts:
\[ \RevB(V) = \sum_{t \in T} f(t) p^*(t) = \sum_{\normone{t} > \capsum} f(t) p^*(t) + \sum_{\normone{t} \le \capsum} f(t) p^*(t). \]

Since $p^*(t) \le b$, the first term is at most $\budget \sum_{\normone{t} > \capsum} f(t) = b \cdot \Prob{}{\normone{v} > \capsum} \le (b/c)\BRevB(V)$, because we can sell the grand bundle at price $p = \capsum$.

The second term is at most $\Rev(\hat V)$, because $M^*$ is a feasible solution to the LP for $\hat T \subseteq T$.
In other words, $M^*$ satisfies the IC and IR constraints for $\hat V$.
The revenue of $\hat V$ is at least the revenue of $M^*$ on $\hat V$:
\[
\Rev(\hat V) \ge \sum_{t \in \hat T} \hat f(t) p^*(t) \ge \sum_{t \in T, \normone{t} \le \capsum} \hat f(t) p^*(t) \ge \sum_{t \in T, \normone{t} \le \capsum} f(t) p^*(t).
\]

Combining the upper bounds, we get
$\RevB(V) \leq  (\budget / \capsum) \BRevB(V) + \Rev(\hat V)$.
\end{proof}

\subsection{Proof of Lemma~\ref{lem:revvp-revbvp}}
\label{sec:revvp-revbvp}
Lemma~\ref{lem:revvp-revbvp} states that when the sum of $v' \sim V'$ is often small, the budget does not matter too much for $V'$.
Intuitively, because each coordinate of $v' \sim V'$ is independent and upper bounded by $b$, a concentration inequality implies that the sum has an exponentially decaying tail.
Therefore, the budget constraint is less critical because it is very unlikely that the buyer's value for the grand bundle is much larger than the budget.

We formalize this intuition by proving the following lemma, which is similar to standard tail bounds.
The main difference is that, instead of knowing the mean of $\normone{v'}$ is small, we only know that $\BRevB(V')$ is small.
\begin{lemma}
\label{lem:sum-vp-tail}
If $V'$ is independent and $\norminf{v'} \le c$ for all $v' \sim V'$, then
\[
\Prob{v' \sim V'}{\normone{v'} \ge x + y + c} \le \Prob{v' \sim V'}{\normone{v'} \ge x} \cdot \Prob{v' \sim V'}{\normone{v'} \ge y} \quad \text{for all } x, y > 0.
\]
In particular, if $\Prob{v' \sim V'}{\normone{v'} \ge c} \le q$, then for all integer $k \ge 0$,
\[
\Prob{v' \sim V'}{\normone{v'} \ge (2k+1) c} \le q^k \Prob{v' \sim V'}{\normone{v'} \ge c}.
\]
\end{lemma}
%
%
We defer the proof of Lemma~\ref{lem:sum-vp-tail} to Appendix~\ref{apx:tail-bound}, and first use this tail bound to prove Lemma~\ref{lem:revvp-revbvp}.

\begin{proof}[of Lemma~\ref{lem:revvp-revbvp}]
Let $c = b$ and $q = \frac{1}{10}$.
We know that $\Prob{}{\normone{v'} \ge c} \le q$ from the assumption $\BRevB(V') \le \frac{\budget}{10}$.

First observe that
$
\BRev^\budget(V') = \max_{p \le \budget} \left(p \cdot \Prob{}{\normone{v'} \ge p}\right).
$
If we price the grand bundle at price $p$ where $(2k+1)c < p \le (2k+3) c$ for some $k \ge 0$,
  by Lemma~\ref{lem:sum-vp-tail}, the revenue is at most
\[
p \cdot \Prob{}{\normone{v'} \ge (2k+1)c} \le (2k+3)c \cdot q^{k} \Prob{}{\normone{v'} \geq c} \le (2k+3)q^{k}\BRevB(V').
\]
For $\SRev(V')$, similar to Lemma~\ref{lem:opt-revvp}, we can upper bound the revenue by allowing the seller to extract full revenue if $\normone{v'} > c$,
  and running the optimal budget-constrained mechanism when $\normone{v'} \le c$:
\begin{align*}
\SRev(V') & \le \SRev^{c}(V') + \expect{}{\normone{v'} \;\middle|\; \normone{v'} \ge c} \\
  & \le \SRev^{c}(V') + \sum_{k=1}^{\infty} (2k+3) c \cdot \Prob{}{\normone{v'} \ge (2k+1)c} \\
  & \le \SRev^{c}(V') + \sum_{k=0}^{\infty} (2k+3)q^{k} \BRev^{c}(V') \\
  & \le \SRev^{c}(V') + 4\BRev^{c}(V'). \tag*{\qed}
\end{align*}
\end{proof}

\subsection{Proof of Lemma~\ref{lem:vp-sbrev}}
\label{sec:vp-sbrev}

Lemma~\ref{lem:vp-sbrev} states that $\SRevB(V)$ and $\BRevB(V)$ are both (up to constant factors) monotone in $V$.
We prove a more general version of the lemma that does not require $V$ to be independent.
Recall that $\hat V \preceq V$ means $\hat V$ is coordinate-wise stochastically dominated by $V$.

\begin{lemma}
\label{lem:vhat-sbrev}
Fix $b > 0$ and $0 < \capsum \le \budget$.
For any distribution $\hat V \preceq V$, $\BRev^c(\hat V) \le \BRevB(V)$ and $\SRev^c(\hat V) \le \max\left(1, \frac{2\capsum}{\budget}\right)\SRevB(V)$.
\end{lemma}

Lemma~\ref{lem:vp-sbrev} follows directly from Lemma~\ref{lem:vhat-sbrev}, by choosing $c = b$ and $\hat V = V'$.

Intuitively, we would like to prove that $\SRevB(\hat V) \le \SRevB(V)$ for any $\hat V \preceq V$.
While this is true in the budget-free case (See Section~\ref{sec:prelim-preceq}), it is actually false in the presence of a budget.
We give a counterexample in Appendix~\ref{apx:srevb-not-monotone}.
Fortunately, we can prove $\SRevB(V') \le 2\SRevB(V)$.  The intuition is that we can cap the price of each item at $\budget/2$, then the buyer either spend at least $\budget/2$, or he will purchase everything he likes.

\begin{proof}[of Lemma~\ref{lem:vhat-sbrev}]
First consider $\BRev$. Because $\capsum \le \budget$ and $\hat V \preceq V$,
\[
\BRev^{\capsum}(\hat V) = \max_{p \le \capsum} \left(p\cdot \Prob{\hat v \sim \hat V}{\normone{\hat v} \ge \capsum}\right) \le \max_{p \le \budget} \left(p\cdot \Prob{v \sim V}{\normone{v} \ge \budget}\right) = \BRevB(V).
\]

For $\SRev$,
  let $\hat M$ be the optimal mechanism that achieves $\SRev^c(\hat V)$ by pricing each item separately.
We construct a mechanism $M$ to mimic $\hat M$ except the prices are capped at $\budget / 2$.
Consider applying $M$ to a buyer with valuation drawn from $V$ and a budget $\budget$.
As $\hat V \preceq V$, we can couple the realizations $\hat v \sim \hat V$ and $v \sim V$ such that $\hat v \leq v$ (coordinate-wise).
For every $(\hat v, v)$ pair:
\begin{itemize}
\item If $M$ gets a revenue of at least $\frac{\budget}{2}$ on $v$. This is at least $\frac{\budget}{2\capsum}$-fraction of the revenue $\hat M$ gets on $\hat v$, because the latter is at most $\capsum$.
\item If $M$ gets a revenue less than $\frac{\budget}{2}$ on $v$, then the buyer has enough budget left to buy any item. 
Therefore, the buyer can buy everything he wants. Because $\hat v \leq v$, the revenue of $M$ on $v$ is at least that of $\hat M$ on $\hat v$.
\end{itemize}
Thus, $M$ can get at least $\min(1, \frac{\budget}{2\capsum})$-fraction of the revenue that $\hat M$ gets on $\hat v$, which implies $\SRev(\hat V) \le \max\left(1, \frac{2\capsum}{\budget}\right) \SRevB(V)$.
\qed
\end{proof}


\section{Public Budget and Weakly Correlated Valuations}
\label{sec:overview}
In this section, we present an alternative approach to prove our main result (Theorem~\ref{thm:main-informal}).
Recall that the buyer has a public budget $b$ and valuations drawn from an independent distribution $V$.

\thmskip
{\noindent \bf Theorem~\ref{thm:main-informal}.~}
{\em 
$\RevB(V) \le 5 \SRevB(V) + 6 \BRevB(V)$.
}

\paragraph{\bf Overview of Our Approach.}
We will truncate the input distribution $V$ in a different way: instead of truncating $v \sim V$ in $\ell_\infty$-norm (as in Section~\ref{sec:proof1}), we will truncate $v$ in $\ell_1$-norm.
This truncation produces a correlated distribution $\hat V$.
The upshot of truncating in $\ell_1$-norm is that we always have $\normone{\hat v} \le \budget$, so $\hat V$ can ignore the budget. 
In addition, as in Section~\ref{sec:proof1}, we can relate the optimal revenue to the revenue of $\hat V$ (Lemma~\ref{lem:revbvsplit}), and we can relate the revenue of simple mechanisms on $\hat V$ back to revenue of simple mechanisms on $V$ (Lemma~\ref{lem:vhat-sbrev}).

We still need to argue that simple mechanisms work well for $\hat V$.
This is the main challenge in this approach.
Because $\hat V$ is correlated, we cannot apply the result of~\cite{BabaioffILW14} in a black-box way.
Instead, we need to modify the analysis of previous work~\cite{LiY13,BabaioffILW14,CaiDW16} and build on the key ideas like ``core-tail''  decomposition.
More specifically, we generalize the duality-based framework developed in~\cite{CaiDW16} to handle the specific type of correlation $\hat V$ has.

%

\paragraph{\bf Weakly Correlated Valuations.}
It is worth mentioning that our structural lemmas (Lemmas~\ref{lem:revbvsplit}~and~\ref{lem:vhat-sbrev}) 
  do not require the input distribution to be independent.
This is why our techniques can be applied to more general settings.
For example, in this section, we will generalize Theorem~\ref{thm:main-informal} with minimum effort to handle weakly correlated valuations (see Definition~\ref{def:weakly} for the formal definition).

\thmskip
{\noindent \bf Corollary~\ref{cor:weakly}.~}
{\em
Let $\hat V$ be a weakly correlated distribution (Definition~\ref{def:weakly}).
We have $\RevB(\hat V) \le 5 \SRevB(\hat V) + 6 \BRevB(\hat V)$.
}
\thmskip

Our main contribution in this section is Lemma~\ref{lem:hatv-simple}.
Lemma~\ref{lem:hatv-simple} shows that for any weakly correlated distribution $\hat V$ (see Definition~\ref{def:weakly}), the better of $\SRev(\hat V)$ and $\BRev(\hat V)$ is a constant approximation to the optimal revenue $\Rev(\hat V)$.

%

\begin{lemma}
\label{lem:hatv-simple}
Fix $c > 0$. Let $\hat V = V_{|(\normone{v} \le \capsum)}$ for an independent distribution $V$.
We have $\Rev(\hat V) \le 5 \SRev(\hat V) + 4\BRev(\hat V)$.
\end{lemma}

We defer the proof of Lemma~\ref{lem:hatv-simple} to Appendix~\ref{sec:vhat}.
We first use these lemmas to prove Theorem~\ref{thm:main-informal} and Corollary~\ref{cor:weakly}.

\begin{proof}[of Theorem~\ref{thm:main-informal} and Corollary~\ref{cor:weakly}]
If $\min_{v \sim V} \normone{v} \ge \budget / 2$, then the seller can price the grand bundle at $\budget / 2$ and the buyer always buys it.
In this case, the revenue is $\budget/2$ and $\RevB(V) \le \budget \le 2 \BRevB(V)$.
Thus, we focus on the more interesting case where $\Prob{v \sim V}{\normone{v} \leq \budget / 2} > 0$.~\footnote{Throughout the paper, when we consider the conditional distribution $\hat V := V_{|(\normone{v} \le \capsum)}$, we will always have $\capsum > \min_{v \in \supp(V)} \normone{v}$, so that the event we condition on happens with non-zero probability.}

Let $\hat V := V_{|(\normone{v} \leq \capsum)}$ for $\capsum = \budget / 2$.
We will reuse Lemmas~\ref{lem:revbvsplit}~and~\ref{lem:vhat-sbrev} from Section~\ref{sec:proof1}.
We can reuse both lemmas because they do not require $V$ or $\hat V$ to be independent, $\hat V$ does not modify the small-sum part of $V$, and $\hat V \preceq V$ (which we will prove as Lemma~\ref{CLM:HATV-PREC-V} in Appendix~\ref{apx:hatv-preceq-v}).
\begin{align*}
\RevB(V) & \le (\budget / \capsum) \cdot \BRevB(V) + \Rev(\hat V) \tag{Lemma~\ref{lem:revbvsplit}}\\
  & \le 2 \BRevB(V) + 5 \SRev(\hat V) + 4\BRev(\hat V) \tag{Lemma~\ref{lem:hatv-simple}} \\
  & = 2 \BRevB(V) + 5 \SRev^\capsum(\hat V) + 4\BRev^\capsum(\hat V) \tag{$\normone{\hat v} \le \capsum$} \\
  & \le 2 \BRevB(V) + 5 \SRevB(V) + 4\BRevB(V). \tag{Lemma~\ref{lem:vhat-sbrev}}
\end{align*}

We now prove Corollary~\ref{cor:weakly}.
Intuitively, Corollary~\ref{cor:weakly} holds because simple mechanisms work well for weakly correlated valuations, and the 
  the weakly-correlated notion is closed under further capping the sum.
  
Let $V = V_{|(\normone{v} \le \capsum_2)}$ be the input distribution.
If $c_2 \le b$, then we can remove the budget constraint and apply Lemma~\ref{lem:hatv-simple} directly.
If $c_2 > b$, then we can cap $V$ at $c_1 = b/2$ to obtain a weakly correlated distribution $\hat V$.
One can verify that Lemmas~\ref{lem:revbvsplit}~and~\ref{lem:vhat-sbrev} still hold for $V$ and $\hat V$, and Lemma~\ref{lem:hatv-simple} holds for $\hat V$.
The only difference is that we need to show $V_{|(\normone{v} \le \capsum_1)} \preceq V_{|(\normone{v} \le \capsum_2)}$ for $\capsum_1 \le \capsum_2$.
We will prove this (Lemma~\ref{CLM:C1-PREC-C2}) in Appendix~\ref{apx:hatv-preceq-v}.
\qed
\end{proof}

\section{Private Budget}
\label{sec:private}
In this section, we consider the case where the budget $b$ is no longer fixed but instead drawn from a distribution $B$.
One natural model is that the buyer's budget $b$ is first drawn from $B$, and then depending on the value of $b$, the buyer's valuations are drawn independently for each item.

We show that in this case, the problem is at least as hard as finding (approximately) optimal mechanisms for correlated valuations in the budget-free setting.
Consider an instance in which all possible budgets are larger than $\max_{v \sim V} \normone{v}$ so they are irrelevant.
However, the budget can still be used as a signal (or a correlation device) to produce correlated valuations.
It is known that for correlated distributions, the better of selling separately and bundling together~\cite{HartN12}, or even the best partition-based mechanism~\cite{BabaioffILW14}, does not offer a constant approximation.

This negative result motivates us to study the private budget setting when the budget distribution $B$ is independent of the valuation distributions $V$.

%

\subsection{Monotone-Hazard-Rate Budgets}
We focus on the case where the budget is independent of valuations, and it is drawn from a continuous\footnote{If the distribution is a discrete MHR distribution, 
  similar results still hold. For discrete distributions we have $\Pr_{b \sim B}[b \geq \lfloor b^* \rfloor] \geq e^{-1}$ instead of $\Pr_{b \sim B}[b \geq b^*] \geq e^{-1}$.}
  monotone-hazard-rate (MHR) distribution. Let $g(\cdot)$ and $G(\cdot)$ be the probability density function and cumulative distribution function of $B$. The MHR condition says $\frac{g(b)}{1 - G(b)}$ is non-decreasing in $b$.

\begin{lemma}
\label{lem:private_MHR}
Let $b^*$ be the expectation of an MHR distribution $B$.
Let $M^*$ be the optimal mechanism for a buyer with a public budget $b^*$.
Then in expectation, $M^*$ extracts at least $\frac{1}{2e}$-fraction of the expected optimal revenue when the buyer has a private budget drawn from $B$.
\end{lemma}
\begin{proof}
Let $R(b, V)$ denote the expected revenue of $M^*$ when the buyer has a public budget $b$ and valuations drawn from $V$. 
Let $R(B, V) = \expect{b \sim B}{R(b, V)}$ denote the expected revenue of $M^*$ when the buyer's budget is drawn from $B$.
\begin{align*}
R(B, V) &= \int_b g(b) R(b,V) \mathrm{d} b
  \ge \int_{b \geq b^*} g(b) R(b, V) \mathrm{d} b \\
  &= \int_{b \geq b^*} g(b) R(b^*, V) \mathrm{d} b \ge e^{-1} \cdot R(b^*, V).
\end{align*}
The second last step uses $R(b, V) = R(b^*, V)$ when $b \geq b^*$, because $M^*$ provides a menu of allocation/payment pairs for the buyer to choose from;
A buyer with budget $b \ge b^*$ can afford any option on the menu so he will choose the same option as if he had budget $b^*$.
The last inequality comes from the fact that for any MHR distribution $B$, $\Pr_{b \sim B}[b \geq b^*] \geq e^{-1}$ (see, e.g.,~\cite{barlow1965tables}).

Let $\Rev^B(V)$ denote the optimal revenue we can extract when the buyer has private budgets drawn from $B$.
\begin{align*}
\Rev^B(V) &\leq 
  \int_{b < b^*} g(b) \Rev^{b}(V) \mathrm{d} b + \int_{b \geq b^*} g(b) \Rev^{b}(V) \mathrm{d} b\\
&\leq \int_{b < b^*} g(b) \Rev^{b^*}(V) \mathrm{d} b + \int_{b \geq b^*} g(b) \cdot \frac{b}{b^*} \cdot \Rev^{b^*}(V) \mathrm{d} b\\
&\leq \Rev^{b^*}(V) + \frac{\int_{b} g(b) b \mathrm{d} b}{b^*} \cdot \Rev^{b^*}(V)
  = 2\Rev^{b^*}(V).
\end{align*}
The first line is because the seller can only do better if she knows the buyer's budget $b$.
The third line is because $b^* = \expect{}{b}$.
The second line uses the fact that $\Rev^b(V) \le \Rev^{b^*}(V)$ when $b < b^*$ and $\Rev^b(V) \le \frac{b}{b^*} \Rev^{b^*}(V)$ when $b > b^*$.

We have $\Rev^b(V) \le \Rev^{b^*}(V)$ when $b < b^*$ because a buyer with budget $b^*$ can afford all options from the menu that achieves $\Rev^b(V)$.
When $b > b^*$, consider the menu that achieves $\Rev^b(V)$ and cap all prices at $b^*$.
A buyer with budget $b > b^*$ either chooses the same option as if he had budget $b^*$, or chooses a different option whose price must be $b^*$, and therefore $\Rev^b(V) \le \frac{b}{b^*} \Rev^{b^*}(V)$.

By definition $R(b^*, V) = \Rev^{b^*}(V)$. Therefore, $R(B, V) \geq \frac{1}{2e} \Rev^{B}(V)$.
\qed
\end{proof}

{\noindent \bf Theorem~\ref{thm:private-mhr}.~}
{\em
When the budget distribution $B$ is MHR, the better of pricing items separately and bundling them together achieves a constant fraction of the optimal revenue.
}
\begin{proof}
By pretending the budget is $b^*$,
\begin{align*}
\SRev^B(V) & \geq \int_{b \geq b^*} g(b) \SRev^{b^*}(V) \mathrm{d} b
  \geq \frac{1}{e} \SRev^{b^*}(V).
\end{align*}
Similarly, $\BRev^B(V) \geq \frac{1}{e} \BRev^{b^*}(V)$.
Therefore, by Theorem~\ref{thm:main-informal} and Lemma~\ref{lem:private_MHR}, $\SRev^B(V) + \BRev^B(V) = \Omega(\SRev^{b^*}(V) + \BRev^{b^*}(V)) = \Omega(\Rev^{b^*}(V)) = \Omega(\Rev^B(V))$.
\qed
\end{proof}


\section{Conclusion and Future Directions}
\label{sec:conclusion}
In this paper, we investigated the effectiveness of simple mechanisms in the presence of budgets, and showed that for an additive buyer with independent valuations and a public budget, either selling separately or selling the grand bundle gives a constant approximation to optimal revenue.

The area of designing simple and approximately optimal auctions with budget constraints is still largely unexplored.
Our work leaves many natural follow-up questions.
We only considered selling to a single buyer.
An immediate open question is whether our results can be extended to multiple bidders.
A generalization to multiple bidders is known in the budget-free case~\cite{Yao15,CaiDW16}.

\thmskip
{\noindent \bf Question 1.~}
{\em
Is there a simple mechanism that is approximately optimal for multiple additive buyers, when each buyer has the same public budget $b$?
}
\thmskip

For private budgets where the budget is independent of the valuations, we showed that if the budget distribution satisfies monotone hazard rate, then we can extract a constant fraction of the revenue.
The general case with arbitrary budget distributions appears to be nontrivial and is an interesting avenue for future work.

\thmskip
{\noindent \bf Question 2.~}
{\em
Is there a simple mechanism that is approximately optimal for an additive buyer with private budgets, when the budget distribution is independent of the valuations?
}

%


\paragraph{\bf Acknowledgements.}
Yu Cheng is supported by NSF grants CCF-1527084, CCF-1535972, CCF-1637397, CCF-1704656, IIS-1447554, and NSF CAREER Award CCF-1750140. Kamesh Munagala is supported by NSF grants CCF-1408784, CCF-1637397, and IIS-1447554; and by an Adobe Data Science Research Award. Kangning Wang is supported by NSF grants CCF-1408784 and CCF-1637397.


\bibliographystyle{acm}
\bibliography{names,conferences,main}

\newpage
\appendix
\section{Proof of the Concentration Lemma in Section~\ref{sec:revvp-revbvp}}
\label{apx:tail-bound}
In this section, we prove Lemma~\ref{lem:sum-vp-tail}.
We first restate it for convenience.

\thmskip
{\noindent \bf Lemma~\ref{lem:sum-vp-tail}.~}
{\em
If $V'$ is independent and $\norminf{v'} \le c$ for all $v' \sim V'$, then
\[
\Prob{v' \sim V'}{\normone{v'} \ge x + y + c} \le \Prob{v' \sim V'}{\normone{v'} \ge x} \cdot \Prob{v' \sim V'}{\normone{v'} \ge y} \quad \text{for all } x, y > 0.
\]
In particular, if $\Prob{v' \sim V'}{\normone{v'} \ge c} \le q$, then for all integer $k \ge 0$,
\[
\Prob{v' \sim V'}{\normone{v'} \ge (2k+1) c} \le q^k \Prob{v' \sim V'}{\normone{v'} \ge c}.
\]
}
\begin{proof}
Consider the probability of $\normone{v'} \ge x + y + c$ conditioned on $\normone{v'} \ge x$.
We will show this probability is at most the probability of $\normone{v'} \ge y$.

For every $v'$ with $\normone{v'} \ge x$, there is a unique $j \in [m]$ where $\sum_{i=1}^{j-1} v'_i < x$ but $\sum_{i=1}^j v'_i \ge x$.
Now $v'_j$ is at most $c$, so for the total sum to be at least $x+y+c$, the remaining sum $\sum_{i=j+1}^m v'_i$ must be at least $y$.
Due to the independence of $V'$, this probability is the same conditioned on any values of $(v'_1, \ldots, v'_j)$.
Formally,
\begin{align*}
\Prob{}{\normone{v'} \ge x + y + c} & = \sum_{j=1}^m \Prob{}{\normone{v'} \ge x + y + c \;\wedge\; \sum_{i=1}^{j-1} v'_i < x \;\wedge\; \sum_{i=1}^j v'_i \ge x} \\
& \le \sum_{j=1}^m \Prob{}{\sum_{i=j+1}^m v'_i \ge y \;\wedge\; \sum_{i=1}^{j-1} v'_i < x \;\wedge\; \sum_{i=1}^j v'_i \ge x} \\
& = \sum_{j=1}^m \Prob{}{\sum_{i=j+1}^m v'_i \ge y} \cdot \Prob{}{\sum_{i=1}^{j-1} v'_i < x \;\wedge\; \sum_{i=1}^j v'_i \ge x} \\
& = \Prob{}{\normone{v'} \ge y} \cdot \Prob{}{\normone{v'} \ge x}.
\end{align*}

The second statement can be proved inductively using the first statement.
The inductive step chooses $x = c$ and $y = (2k-1) c$.
\qed
\end{proof}

\section{Revenue Non-monotonicity for Separate Selling to a Budgeted Buyer}
\label{apx:srevb-not-monotone}
We provide an example where $V_1 \preceq V_2$ but $\SRevB(V_1) > \SRevB(V_2)$.
Intuitively, a budget-constrained buyer solves a $\mathsc{Knapsack}$ problem when deciding which items to purchase, and the total volume (i.e, payment) of the optimal $\mathsc{Knapsack}$ solution is not monotone in the item values.
Increasing the value of a cheap item might incentivize the buyer to purchase this item instead of a more expensive one, if the buyer does not have enough budget to buy both items.

Consider an auction with $3$ items.
$T_1$ and $T_2$ are two matrices defined as
\begin{align*}
T_1 = \begin{bmatrix}
       2 & 0 & 0 \\
       0 & 1 & 1 \\
       2 & 1 & 0 \\
     \end{bmatrix}, \quad
T_2 = \begin{bmatrix}
       2 & 0 & 0 \\
       0 & 1 & 1 \\
       2 & 2 & 0 \\
     \end{bmatrix}.
\end{align*}
The rows of $T_1$ are the support of $V_1$, and similarly, the rows of $T_2$ are the support of $V_2$.
Each row is associated with a probability of $\frac{1}{3}$.
Assume the budget $b = 2$.

One of the optimal mechanisms that obtains $\SRevB(V_1) = 2$ is to price the items at $(2,1,1)$.
A buyer of type $1$ and $3$ will buy the first item, and a buyer of type $2$ will buy the last two items.
This is optimal because $\SRevB$ cannot exceed the budget.

However, $\SRevB(V_2) < 2$. To prove it, we first notice $\SRevB(V_2) \leq 2$ because $b = 2$. It means we must get a revenue of $2$ from all buyer types to make $\SRevB(V_2) = 2$. Thus, we must price the first item at $2$, and each of the last two items at $1$ to satisfy this constraint for the first two types. Nevertheless, with this pricing strategy, we only get revenue $1$ for the last buyer type $(2, 2, 0)$, because he will only buy the second item.
This shows $\SRevB(V_2) < 2 = \SRevB(V_1)$.




\section{Simple Mechanisms for Weakly Correlated Valuations}
\label{sec:vhat}
This section is devoted to proving Lemma~\ref{lem:hatv-simple}.
Lemma~\ref{lem:hatv-simple} states that simple mechanisms (more specifically, the better of $\SRev(\hat V)$ and $\BRev(\hat V)$) are approximately optimal for any weakly correlated distribution $\hat V = V_{|(\normone{v} \le c)}$.
Note that there is no budget in this section, only a cap $\capsum$ on the $\ell_1$-norm of $\hat v \sim \hat V$.

Our approach builds on the ideas like ``core-tail decomposition'' from previous works that show $\Rev(V) = O(\SRev(V) + \BRev(V))$ for independent valuations $V$~\cite{HartN17,LiY13,BabaioffILW14,CaiDW16}.
More specifically, we generalize the duality-based framework developed in~\cite{CaiDW16} to handle weakly correlated distributions.
The idea of~\cite{CaiDW16} is to Lagrangify only the incentive constraints, then guess the Lagrangian multipliers to derive an upper bound on the maximum revenue.

We first highlight some of the difficulties in extending previous works to correlated distributions.
\begin{enumerate}
\item When the distribution is independent, one can upper bound the maximum revenue by Myerson's \emph{virtual value} of the bidder's favorite item, plus the sum of the \emph{values} of the remaining items. For correlated distributions, the virtual value of the favorite item depends on the other items. 
\item In the core part of core-tail decomposition, we need the total value of the low-value items to concentrate around its expectation, so we can upper bound their values by $\BRev$ (by selling the grand bundle at a price slightly lower than that expectation).
When the valuations are independent, we can show the variance is small, which may not be true for correlated distributions.
\end{enumerate}

These difficulties are not surprising, because Lemma~\ref{lem:hatv-simple} cannot hold for arbitrary correlated distributions.
As shown in~\cite{BriestCKW10,HartN13}, for correlated distributions, the gap between the best deterministic and randomized mechanisms can be unbounded. 
Hence, we have to take advantage of the special properties of $\hat V$.


\paragraph{\bf Notations.}
In this Section, because $\hat V = V_{|(\normone{v} \le c)}$ is the distribution we focus on, we use $T$ to denote the support of $\hat V$.
Given a (correlated) distribution $D$, we use $D_j$ to denote $D$'s marginal distribution on the $j$-th coordinate, and $D_{-j}$ to denote $D$'s marginal (joint) distribution on coordinates other than $j$.
Let $T_j$ and $T_{-j}$ be the support of $\hat V_j$ and $\hat V_{-j}$ respectively.
In addition, we will make use of the conditional distributions $\hat V_{j | \hat v_{-j} = t_{-j}}$ and $\hat V_{-j | \hat v_{j} = t_{j}}$;
The former is the distribution of $\hat v_j$ for $\hat v \sim \hat V$ conditioned on $\hat v_{-j} = t_{-j}$, and the latter is the distribution of $\hat v_{-j}$ conditioned on $v_j = t_j$.
Abusing notation, we use $f(t)$, $f(t_j)$, $f(t_{-j})$, $f(t_j | t_{-j})$, and $f(t_{-j} | t_j)$ to denote the probability mass function of the correlated and conditional distributions we mentioned in this section.
When the value of item $j$ is drawn from $f(t_j | t_{-j})$, we use $\vVALt(t_j | t_{-j})$ to denote item $j$'s (ironed) Myerson's virtual value.

For a bidder type $t \in T \subseteq \R^m$, the favorite item of $t$ is the one with the highest value (with ties broken lexicographically). 
We write $t \in R_j$ if and only if $j$ is the favorite item of type $t$ after tie-breaking.
Formally,
\[
R_j = \{t \in T \mid j \text{ is the smallest index with } t_j \ge t_k \text{ for all } k \in [m]\}.
\]

\paragraph{\bf Proof of Lemma~\ref{lem:hatv-simple}.}
We extend the duality framework in~\cite{CaiDW16} to handle correlated distributions.
As we will see in Section~\ref{sec:canon-flow}, we can upper bound the optimal revenue of $\hat V$ into three components.
Recall that $\pi$ is the allocation rule.
For notational convenience, we write $r$ for $\SRev(\hat V)$.
Notice that in $\SINGLE$ we get $\vVALt(t_j | t_{-j})$ rather than $\vVALt(t_j)$. 
\begin{align*}
\Rev(\hat V) & \le \sum_{t \in T} f(t) \sum_{j \in [m]} \pi_j(t) \vVALt(t_j | t_{-j}) \cdot \indic{t \in R_j} \tag{\SINGLE} \\
  & \quad + \sum_{t \in T} f(t) \sum_{j \in [m]} \pi_j(t) t_j \cdot \indic{t \notin R_j} \cdot \indic{t_j \le 2r} \tag{\CORE} \\
  & \quad + \sum_{t \in T} f(t) \sum_{j \in [m]} \pi_j(t) t_j \cdot \indic{t \notin R_j} \cdot \indic{t_j > 2r}. \tag{\TAIL}
\end{align*}

Lemma~\ref{lem:hatv-simple} follows directly from the statements of Lemmas~\ref{lem:vhat-single},~\ref{lem:vhat-tail},~and~\ref{lem:vhat-core}.

\begin{lemma}
\label{lem:vhat-single}
$\SINGLE \le 2 \BRev(\hat V) + \SRev(\hat V)$.
\end{lemma}

\begin{lemma}
\label{lem:vhat-tail}
$\TAIL \le \SRev(\hat V)$.
\end{lemma}

\begin{lemma}
\label{lem:vhat-core}
$\CORE \le 2 \BRev(\hat V) + 3 \SRev(\hat V)$.
\end{lemma}


\paragraph{\bf Organization.} For completeness, we first recall the approach in \cite{CaiDW16} in Section~\ref{sec:stoc16-duality}.
In Section~\ref{sec:canon-flow}, we show that there exists a choice of the Lagrangian multipliers that $\Rev(\hat V)$ can be upper bounded by $\SINGLE + \CORE + \TAIL$.
Lemmas~\ref{lem:vhat-single},~\ref{lem:vhat-tail},~and~\ref{lem:vhat-core} are proved in Section~\ref{sec:vhat-proofs}.

\subsection{The Duality Based Approach in~\cite{CaiDW16}.}
\label{sec:stoc16-duality}

The optimal mechanism $M^* = (\pi^*, p^*)$ for $\hat V$ is the optimal solution to the following LP:
\begin{lp}
\label{lp:rev-reduced}
\maxi{\sum_{t \in T} f(t) p(t)}
\st \qcon{\pi(t')^\top t - p(t') \le \pi(t)^\top t - p(t)}{\forall t \in T, t' \in T^+}
    \qcon{0 \le \pi_j(t) \le 1}{\forall t \in T, j \in [m]}
    \con{\pi(\varnothing) = 0, \; p(\varnothing) = 0.}
\end{lp}

We can upper bound the optimal primal value by Lagrangifying the incentive constraints.
\begin{align*}
\Rev(\hat V) = \min_{\lambda \ge 0} \max_{\pi, p} L(\lambda, \pi, p),
\end{align*}
where the Lagrangian dual of LP~\eqref{lp:rev-reduced} is given by
\begin{align*}
\begin{split}
L(\lambda,\pi,p) 
  & = \sum_{t \in T} f(t) p(t) + \sum_{t \in T, t' \in T^+} \lambda(t, t') \left[ (\pi(t)-\pi(t'))^\top t - (p(t) - p(t')) \right] \\
  & = \sum_{t \in T} p(t) \left[f(t) - \sum_{t' \in T^+} \lambda(t, t') + \sum_{t' \in T} \lambda(t', t) \right] \\
    & \qquad + \sum_{t \in T} \pi(t)^\top \left[\sum_{t' \in T^+} \lambda(t, t') t - \sum_{t' \in T} \lambda(t', t) t' \right].
\end{split}
\end{align*}

Because the $p(t)$'s are unconstrained variables, any dual solution with a finite value must have
\begin{align}
\label{eqn:multiplier-flow}
f(t) - \sum_{t' \in T^+} \lambda(t, t') + \sum_{t' \in T} \lambda(t', t) = 0, \; \forall t \in T.
\end{align}
From now on, we restrict our attention to only dual solution with finite values.
We can then simplify $L(\lambda,\pi,p)$ by replacing $\sum_{t' \in T^+} \lambda(t, t')$ with $f(t) + \sum_{t' \in T} \lambda(t', t)$ to get rid of $T^+$:
\begin{align*}
L(\lambda,\pi,p)
  & = \sum_{t \in T} \pi(t)^\top \left[f(t) t + \sum_{t' \in T} \lambda(t', t) t - \sum_{t' \in T} \lambda(t', t) t' \right] \\
  & = \sum_{t \in T} f(t) \pi(t)^\top \left[t - \frac{1}{f(t)} \sum_{t' \in T} \lambda(t', t) (t' - t) \right].
\end{align*}

We write $\VVAL(t)$ as a shorthand for the term in the bracket:
$
\VVAL(t) = t - \frac{1}{f(t)} \sum_{t' \in T} \lambda(t', t) (t' - t).
$
We know that 
$
L(\lambda,\pi,p) = \sum_{t \in T} f(t) \pi(t)^\top \VVAL(t) \ge \Rev(\hat V)
$
  is an upper bound on the revenue of the optimal mechanism.
We can rewrite Equation~\eqref{eqn:multiplier-flow} as
\[
f(t) + \sum_{t' \in T} \lambda(t', t) = \sum_{t' \in T^+} \lambda(t, t') = \sum_{t' \in T} \lambda(t, t') + \lambda(t, \varnothing), \; \forall t \in T.
\]
\cite{CaiDW16} interpreted these constraints as flow conservation constraints.
Let $\lambda(t, t') \ge 0$ denote the amount of flow $t$ sends to $t'$.
The left-hand side is the total flow received by $t$, where every type $t$ receives $f(t)$ units of flow from the source;
and the right-hand side is the total flow send out from $t$, with all the excess flow sent to the sink ($\varnothing$).

They proposed a ``canonical flow'' which was shown to be a good guess for the Lagrangian multipliers.
It turns out the same dual solution is sufficient to prove our results for correlated distributions.
In the next section, we recall this canonical flow and use it to derive an upper bound on the optimal revenue.

\subsection{Canonical Flow for Weakly Correlated Distributions}
\label{sec:canon-flow}
Recall that $t \in R_j$ if and only if $j$ is the favorite item of type $t$.
Formally, there exists $\lambda(t,t') \ge 0$ such that
\begin{itemize}
\item For every $j$, all flows entering $R_j$ are from the source, and all flows leaving $R_j$ are to $\varnothing$.
\item For $t, t' \in R_j$, we can have $\lambda(t',t) > 0$ only if $t$ and $t'$ only differ on the $j$-th coordinate. When there is no ironing, $\lambda(t',t) > 0$ only if $t'_j$ is the smallest value larger than $t_j$ in $T_j$.
\end{itemize}

\begin{lemma}
\label{lem:canonical-flow}
There exists a set of the Lagrangian multipliers $\lambda$ that satisfies the flow conservation constraints, such that
\begin{enumerate}
\item[(1)] If $t \notin R_j$, then $\VVAL_j(t) = t_j$.
\item[(2)] If $t \in R_j$, then $\VVAL_j(t) \le \vVALt_j(t_j | t_{-j})$, where $\vVALt(t_j | t_{-j})$ is item $j$'s (ironed) Myerson's virtual value conditioned on $t_{-j}$.
\end{enumerate}
\end{lemma}
\begin{proof}
Recall that
$
\VVAL(t) = t - \frac{1}{f(t)} \sum_{t' \in T} \lambda(t', t) (t' - t).
$
For (1), assume that $t \in R_k$ for some $k \neq j$.
If $\lambda(t', t) > 0$, it must be the case that $t$ and $t'$ are only different on the $k$-th coordinate, so $(t - t')_j = 0$ and $\VVAL_j(t) = t_j$.

Now we prove (2).
We first consider a canonical flow without ironing.
Fix any $j$ and $t \in R_j$.
If $t_j$ is the largest value in $T_j$, then $\lambda(t', t) = 0$ for all $t'$ and $\VVAL_j(t) = t_j$.
If $t_j$ is not the largest value in $T_j$, $t$ receives flow from the source and exactly one other node $t'$ where $t'_{-j} = t_{-j}$, and $t'_j$ is the smallest value larger than $t_j$ in $T_j$.
The total flow from $t'$ to $t$ includes the flows from the source to all $t^*$ with $t^*_{-j} = t_{-j}$ and $t^*_j > t_j$:
\begin{align*}
\lambda(t', t) & = \sum_{t^* \in T} f(t^*) \cdot \indic{t^*_{-j} = t_{-j} \wedge t^*_j > t_j} \\
   & = f(t_{-j}) \sum_{t^*_j \in T_j, t^*_j > t_j} f(t^*_j | t_{-j}) 
   = f(t_{-j}) \left(1 - F(t_j | t_{-j}) \right).
\end{align*}
Substituting $\lambda(t', t)$ in the expression of $\VVAL(t)$, this implies
\begin{align*}
\VVAL_j(t) & = t_j - \frac{1}{f(t)} \sum_{t^* \in T} \lambda(t^*, t) (t^*_j - t_j) \\
  & = t_j - \frac{1}{f(t)} \lambda(t', t)(t'_j - t_j) \\
  & = t_j - \frac{f(t_{-j}) \left(1 - F(t_j | t_{-j}) \right)}{f(t_{-j}) f(t_{j} | t_{-j})} (t'_j - t_j)
  = \vVAL(t_j | t_{-j}). 
\end{align*}

Finally, we show that the flow can be modified to implement Myerson's ironing procedure.
The analysis on modifying the flow to reflect ironing is given in~\cite{CaiDW16}, and we include it here for completeness.
Suppose there exist two types $t, t' \in R_j$ such that $t_j < t'_j$ but $\VVAL_j(t) > \VVAL_j(t')$.
We can add a cycle of $w$ units of flow between $t$ and $t'$, that is, we increase both $\lambda(t, t')$ and $\lambda(t', t)$ by $w$.
Notice that the resulting flow is still valid, and $\VVAL(t^*)$ for all $t^* \neq t, t'$ remain unchanged.
Moreover, the change does not alter $\VVAL_k(t)$ or $\VVAL_k(t')$ for all $k \neq j$.
The only effect of the change is to increase $\VVAL_j(t)$ by $w(t'_j - t_j) / f(t)$, and decrease $\VVAL_j(t')$ by $w(t'_j - t_j) / f(t')$.
Therefore, we can choose $w$ so that $\VVAL_j(t) = \VVAL_j(t')$ without changing any other virtual values.

Repeating this process allows us to simulate Myerson's ironing procedure.
One technical issue is that we may cut off an ironing interval of $f(t_j | t_{-j})$ because it leaves the region $R_j$.
However, we know that truncating an ironing interval $I$ to $I' \subseteq I$ from below can only decrease the virtual value on $I'$.
This is because the average virtual value on $I'$ is smaller than the average virtual value on $I$, otherwise we would not iron the entire interval $I$ in the first place. \qed
\end{proof}


\subsection{Upper Bounds for $\SINGLE$, $\CORE$, and $\TAIL$}
\label{sec:vhat-proofs}

We decompose the upper bound we had in the previous section into three components.
Recall that $r = \SRev(\hat V)$.
By Lemma~\ref{lem:canonical-flow}, we know that
\begin{align*}
\Rev(\hat V) & \le \sum_{t\in T} f(t) \pi(t)^\top \VVAL(t) \\
  & = \sum_{t \in T} f(t) \sum_{j \in [m]} \pi_j(t) \left(\vVALt(t_j | t_{-j}) \cdot \indic{t \in R_j} + t_j \cdot \indic{t \notin R_j}\right) \\
  & \le \sum_{t \in T} f(t) \sum_{j \in [m]} \pi_j(t) \vVALt(t_j | t_{-j}) \cdot \indic{t \in R_j} \tag{\SINGLE} \\
  & \quad + \sum_{t \in T} f(t) \sum_{j \in [m]} \pi_j(t) t_j \cdot \indic{t \notin R_j} \cdot \indic{t_j \le 2r} \tag{\CORE} \\
  & \quad + \sum_{t \in T} f(t) \sum_{j \in [m]} \pi_j(t) t_j \cdot \indic{t \notin R_j} \cdot \indic{t_j > 2r}. \tag{\TAIL}
\end{align*}

We now restate the lemmas that upper bounds each component.

\thmskip
{\noindent \bf Lemma~\ref{lem:vhat-single}.~}
{\em
$\SINGLE \le 2 \BRev(\hat V) + \SRev(\hat V)$.
}
\begin{proof}
We first recall the expression of \SINGLE.
\[
\SINGLE = \sum_{t \in T} f(t) \sum_{j \in [m]} \pi_j(t) \vVALt(t_j | t_{-j}) \cdot \indic{t \in R_j}.
\]
Because $\pi_j(t)\cdot \indic{t \in R_j}$ is between $0$ and $1$, we can upper bound $\SINGLE$ by setting it to $1$ whenever the ironed virtual value is positive, and $0$ otherwise. 
\begin{align*}
\SINGLE 
  & \le \sum_{t \in T} f(t) \sum_{j \in [m]} \vVALt(t_j | t_{-j}) \cdot \indic{\vVALt(t_j | t_{-j}) \ge 0} \\
  & = \sum_{j \in [m]} \sum_{t \in T} f(t_{-j}) f(t_j | t_{-j}) \vVALt(t_j | t_{-j}) \cdot \indic{\vVALt(t_j | t_{-j}) \ge 0} \\
  & = \sum_{j \in [m]} \sum_{t_{-j} \in T_{-j}} f(t_{-j}) \cdot \Rev(t_j | t_{-j}).
\end{align*}
Intuitively, we need to show that knowing $t_{-j}$ does not help us sell item $j$ by too much.
Observe that $\hat V$ is a correlated distribution obtained from capping an independent distribution.
The proof of the lemma crucially relies on the following property of $\hat V$: revealing $t_{-j}$ gives the same amount of information as revealing only the sum of $t_{-j}$.
For a fixed $j$, the revenue $\Rev(t_j|t_{-j})$ can be captured by the (non-disjoint) union of the following two cases:
\begin{enumerate}
\item[(1)] $\normone{t_{-j}} < \capsum/2$ and $t_j < \capsum/2$ both hold. Conditioned on this event, the valuation of $t_j$ is independent of $t_{-j}$. Hence, knowing $t_{-j}$ does not provide additional information.
\item[(2)] $\normone{t} \ge \capsum/2$. In this case, buyer's value for the grand bundle is at least $\capsum/2$, so we could charge this to $\BRev(\hat V)$.
\end{enumerate}
It is worth noting that $t_j$ and $t_{-j}$ are not independent when $\normone{t} < \capsum/2$, so we have to condition on stricter events for them to become independent. 
Formally, if $\normone{t_{-j}} < \capsum/2$,
\[
f(t_j|t_{-j}, t_j < \capsum/2) = f(t_j|t_j < \capsum/2), \quad \forall j \in [m], t \in T.
\]

We are now ready to bound $\SINGLE$.
For single-parameter distributions, the optimal auction simply sets a reserve price.
Let $p^*_j$ be the optimal reserve price for the distribution $f(t_j|t_j < \capsum/2)$, and let $p_j(t_{-j})$ be the optimal reserve price for the distribution $f(t_j|t_{-j})$.
\begin{align*}
\SINGLE & \le \sum_{j \in [m]} \sum_{t_{-j} \in T_{-j}} f(t_{-j}) \cdot \Rev(t_j | t_{-j}) \\
  & = \sum_{j \in [m]} \sum_{t \in T} f(t_{-j}) f(t_j|t_{-j}) p_j(t_{-j}) \cdot \indic{t_j \ge p_j(t_{-j})} \\
  & \le \sum_{j} \sum_{\normone{t} \ge \frac{\capsum}{2}} f(t) t_j + \sum_{j} \sum_{t_j < \frac{\capsum}{2}, \normone{t_{-j}} < \frac{\capsum}{2}} f(t) p_j(t_{-j}) \cdot \indic{t_j \ge p_j(t_{-j})} \\ 
  & \le \sum_{\normone{t} \ge \frac{\capsum}{2}} f(t) \sum_{j} t_j + \sum_{j} \sum_{t_j < \frac{\capsum}{2}, \normone{t_{-j}} < \frac{\capsum}{2}} f(t) p^*_j \cdot \indic{t_j \ge p^*_j} \\
  & \le \capsum \cdot \Prob{t \sim \hat V}{\normone{t} \ge \frac{\capsum}{2}} + \sum_{t\in T} f(t) \sum_{j\in[m]} p^*_j \cdot \indic{t_j \ge p^*_j} \\
  & \le 2 \BRev(\hat V) + \SRev(\hat V).
\end{align*}
The last step uses the facts that (1) we can price the grand bundle at price $\capsum/2$ and therefore $\BRev(\hat V) \ge (\capsum/2) \cdot \Prob{}{\normone{t} \ge \capsum/2}$; and (2) the second term is exactly the revenue we can obtain if we post each item $j$ separately at price $p^*_j$. \qed
\end{proof}

Recall that $r = \SRev(\hat V)$. We continue to upper bound $\TAIL$ and $\CORE$.

\thmskip
{\noindent \bf Lemma~\ref{lem:vhat-tail}.~}
{\em
$\TAIL \le \SRev(\hat V)$.
}
\begin{proof}
Recall that $R_j \subseteq T$ is the subset of buyer types whose favorite item is $j$.
\begin{align*}
\TAIL & = \sum_{t \in T} f(t) \sum_{j \in [m]} \pi_j(t) t_j \cdot \indic{t \notin R_j} \cdot \indic{t_j > 2r} \\
  & \le \sum_{t \in T} f(t) \sum_{j \in [m]} t_j \cdot \indic{t \notin R_j} \cdot \indic{t_j > 2r} \\
  & = \sum_{j \in [m]} \sum_{t_j > 2r} f(t_j) \sum_{t_{-j} \in T_{-j}} f(t_{-j}|t_j) t_j \cdot \indic{t \notin R_j}.
\end{align*}
We first show that for any fixed $t^*_j > 2r$, knowing the value for item $j$ to be $t^*_j$ will not increase the probability of $t \notin R_j$ by more than a factor of $2$.
Intuitively, it is because $t_j > 2r$ is large enough.
This is another place where we use the special property of the $\hat V$ we designed.

Recall that $\hat V_j$ is the marginal distribution of $\hat V$ on the $j$-th coordinate, and $f(t_j)$ is the probability mass function of $\hat V_j$.
The definition $r = \SRev(\hat V)$ implies that $\Prob{t_j \sim \hat V_j}{t_j < 2r} \ge 1/2$, otherwise selling only item $j$ at price $2r$ gives revenue more than $r$.
\begin{align*}
& \sum_{t_{-j} \in T_{-j}} f(t_{-j}|t^*_j) \cdot \indic{(t^*_j, t_{-j}) \notin R_j} \\
  & \le 2 \Prob{t_j \sim \hat V_j}{t_j < 2r} \cdot \sum_{t_{-j} \in T_{-j}} f(t_{-j}|t^*_j) \cdot \indic{(t^*_j, t_{-j}) \notin R_j} \\
  & = 2 \sum_{t_j < 2r} f(t_j) \sum_{t_{-j} \in T_{-j}} f(t_{-j}|t^*_j) \cdot \indic{(t^*_j, t_{-j}) \notin R_j} \\
  & \le 2 \sum_{t_j < 2r} f(t_j) \sum_{t_{-j} \in T_{-j}} f(t_{-j}|t_j) \cdot \indic{(t_j, t_{-j}) \notin R_j} \\
  & \le 2 \sum_{t_{-j} \in T_{-j}} f(t_{-j}) \cdot \indic{t \notin R_j}.
\end{align*}
The last step uses the fact that $f(t_{-j}) = \sum_{t_j} f(t_j)f(t_{-j}|t_j)$.
The second last step states the event $t \notin R_j$ (that $j$ is not the largest coordinate) is more likely to happen when the value of $t_j$ is smaller. It uses the monotonicity of $f(t_{-j} | \, \cdot \,)$ and $t_j < 2r < t^*_j$.
This fact is captured in Lemma~\ref{CLM:C1-PREC-C2} and will be proved in Appendix~\ref{apx:hatv-preceq-v}.

\thmskip
{\noindent \bf Lemma~\ref{CLM:C1-PREC-C2}.~}
{\em
Let $\hat V = V_{|(\normone{v} \le \capsum_1)}$ and $\tilde V = V_{|(\normone{v} \le \capsum_2)}$ for any $c_1 \leq c_2$. We have $\hat V \preceq \tilde V$.
}

Finally, we relate our upper bound on $\TAIL$ to $r$.
\begin{align*}
\TAIL & \le \sum_{j \in [m]} \sum_{t_j > 2r} f(t_j) \sum_{t_{-j} \in T_{-j}} f(t_{-j}|t_j) t_j \cdot \indic{t \notin R_j} \\
  & \le 2 \sum_{j \in [m]} \sum_{t_j > 2r} f(t_j) \sum_{t_{-j} \in T_{-j}} f(t_{-j}) t_j \cdot \indic{t \notin R_j} \\
  & \le 2 \sum_{j \in [m]} \sum_{t_j > 2r} f(t_j) \cdot r
    \le \SRev(\hat V).
\end{align*}
The second last step is because for any $j$, selling each item separately at the same price $t^*_j$ gives revenue at least $\sum_{t_{-j}} f(t_{-j}) t^*_j \cdot \indic{t \notin R_j}$:
If some item $k \neq j$ satisfies that $t_k \ge t^*_j$, then the buyer would purchase at least one of such items.
The last step holds because $\sum_{j \in [m]} \sum_{t_j \ge 2r} f(t_j) \cdot 2r$ is exactly the revenue of selling each item at $2r$, so it is at most $\SRev(\hat V)$. \qed
\end{proof}

\thmskip
{\noindent \bf Lemma~\ref{lem:vhat-core}.~}
{\em
$\CORE \le 2 \BRev(\hat V) + 3 \SRev(\hat V)$.
}
\begin{proof}
Recall $r = \SRev(\hat V)$.
If we sell the grand bundle at price $\CORE - 3r$, we show that the buyer will purchase it with probability at least $5/9$.
This implies that $\BRev(\hat V) \ge \frac{5}{9}(\CORE - 3r)$, or equivalently
$
\CORE \le \frac{9}{5} \BRev(\hat V) + 3 r \le 2 \BRev(\hat V) + 3 \SRev(\hat V).
$

For the simplicity of presentation, we define a new random variable $\vcap \in \R^m$ as follows: we first draw a sample $\hat v$ from $\hat V$, and set $\vcap_j = \min(\hat v_j, 2r)$ for all $j \in [m]$.
\begin{align*}
\CORE & = \sum_{t \in T} f(t) \sum_{j \in [m]} \pi_j(t) t_j \cdot \indic{t \notin R_j} \cdot \indic{t_j \le 2r} \\
  & \le \sum_{t \in T} f(t) \sum_{j \in [m]} t_j \cdot \indic{t_j \le 2r} \\
  & \le \sum_{t \in T} f(t) \sum_{j \in [m]} \left(t_j \cdot \indic{t_j \le 2r} + 2r \cdot \indic{t_j > 2r} \right) 
    = \expect{}{\normone{\vcap}}.
\end{align*}
Since we price the grand bundle at $\CORE - 3r \le \expect{}{\normone{\vcap}} - 3r$, and the buyer's value for the grand bundle is $\normone{t}$, it is sufficient to prove 
\[
\Prob{t \sim \hat V}{\normone{t} \ge \expect{}{\normone{\vcap}} - 3r} \ge \frac{5}{9}.
\]
Moreover, because $\hat V$ stochastically dominates $\vcap$, 
it is sufficient to prove
\[
\Prob{}{\normone{\vcap} \ge \expect{}{\normone{\vcap}} - 3r} \ge \frac{5}{9}.
\]
Intuitively, the condition says that the $\ell_1$-norm of the random variable $\vcap$ concentrates around its expectation.
This is the crucial reason why we design our $\hat V$ to be negatively correlated. 

In the next claim (Lemma~\ref{clm:vhat-varc}), we are going to prove $\var[\normone{\vcap}] \le 4 r^2$.
Assume this is true, we conclude the proof by applying the Chebyshev inequality,
\[
\Prob{}{\normone{\vcap} < \expect{}{\normone{\vcap}} - 3r} \le \frac{\var{\normone{\vcap}}}{9r^2} \leq \frac{4}{9}. \tag*{\qed} 
\] 
\end{proof}

\begin{lemma}
\label{clm:vhat-varc}
Let $c \in \R^m$ be the random variable with $\vcap_j = \min(\hat v_j, 2r)$ for all $j$ and $\hat v \sim \hat V$.
We have $\var[\normone{\vcap}] \le 4 r^2$.
\end{lemma}
\begin{proof}
We first show that for any $i \neq j$, $\vcap_i$ and $\vcap_j$ are negatively correlated, i.e., 
\[
\cov(\vcap_i, \vcap_j) = \mathbb{E}[\vcap_i \vcap_j] - \expect{}{\vcap_i} \cdot \mathbb{E}[\vcap_j] \le 0.
\]

Observe that $\hat V_{j|\hat v_i = x}$ stochastically dominates $\hat V_{j|\hat v_i = y}$ for any two possible values $x < y$ of $\hat v_i$.
This is because, for any $a \in \R$,
\begin{align*}
& \Prob{\hat v \sim \hat V}{\hat v_j \le a \ | \ \hat v_i = x} \\
  & = \Prob{v \sim V}{v_j \le a \ \left| \ v_i = x, \normone{v} \le \capsum\right.} \\
  & = \Prob{v_j \sim V_j}{v_j \le a \ \left| \ v_j \le \capsum - x - \sum_{k\neq i,j} v_k\right.} \\
  & \le \Prob{v_j \sim V_j}{v_j \le a \ \left| \ v_j \le \capsum - y - \sum_{k\neq i,j} v_k\right.} 
    = \Prob{\hat v \sim \hat V}{\hat v_j \le a \ \left| \ \hat v_i = y\right.}.
\end{align*}
Because $\vcap_j = \min(\hat v_j, 2r)$, we know that $\vcap_{j | v_i = x} \succeq \vcap_{j | v_i = y}$ as well.
It follows that $\mathbb{E}[\vcap_j | v_i = x] \ge \mathbb{E}[\vcap_j | v_i = y]$.
In addition, since $\mathbb{E}[\vcap_j | \vcap_i = x] = \mathbb{E}[\vcap_j | v_i = x]$ when $x < 2r$, and $\mathbb{E}[\vcap_j | \vcap_i = 2r] = \mathbb{E}[\vcap_j | v_i \ge 2r]$, we can deduce that $\mathbb{E}[\vcap_j|\vcap_i]$ weakly decreases as $\vcap_i$ increases. Therefore, we get
\[
\mathbb{E}[\vcap_i \vcap_j] = \sum_{\vcap_i} \Prob{}{\vcap_i} \left(\vcap_i \cdot \mathbb{E}[\vcap_j | \vcap_i] \right)
  \le \sum_{\vcap_i} \Prob{}{\vcap_i} \left(\vcap_i \cdot \mathbb{E}[\vcap_j]\right) = \expect{}{\vcap_i} \cdot \mathbb{E}[\vcap_j],
\]
by an application of (a generalization of) the rearrangement inequality.

Given the negative correlations between the $\vcap_i$'s, we can upper bound the variance of their sum.
\[
\var\left(\sum_{i \in [m]} \vcap_i\right) = \sum_{i \in [m]} \var(\vcap_i) + 2 \sum_{1 \le i < j \le m} \cov(\vcap_i, \vcap_j) \le \sum_{i \in [m]} \var(\vcap_i).
\]
It remains to show $\sum_{i} \var(\vcap_i) \le 4 r^2$.
This part is identical to the analysis in earlier works~\cite{LiY13,BabaioffILW14,CaiDW16}, but we include it for completeness.
Let $r_j \in \R$ denote the maximum revenue one can extract by selling item $j$ alone.
Note that $\SRev(\hat V) = r = \sum_{j} r_j$, so it is sufficient to show $\var(\vcap_j) \le 4 r r_j$ for all $j$.

Fix some $j \in [m]$.
We use $x = \vcap_j$ as an alias for the random variable $\vcap_j$.
We know that
\begin{enumerate}
\item[(1)] $x$ is at most $2r$, and
\item[(2)] the revenue of $x$ is at most $r_j$: $\left(a \cdot \Prob{}{x \ge a}\right) \le r_j$ for any $a \in \R$.
\end{enumerate}
Combining these two facts gives the required bound on the variance of $x$.
Let $0 < a_1 < \ldots < a_\ell \le 2r$ be the support of $x$, and let $a_0 = 0$.
\begin{align*}
\mathbb{E}[x^2] & = \sum_{k=1}^\ell \Prob{}{x = a_k} \cdot a_k^2 \\
  & = \sum_{k=1}^\ell (a_k^2 - a_{k-1}^2) \cdot \Prob{}{x \ge a_k} \\
  & < \sum_{k=1}^\ell 2 (a_k - a_{k-1}) \left(a_k \cdot \Prob{}{x \ge a_k}\right) \\
  & \le r_j \sum_{k=1}^\ell 2 (a_k - a_{k-1}) \\
  & = 2 r_j a_\ell \le 4 r r_j. \tag*{\qed} 
\end{align*}
\end{proof}

\section{Proof of Lemma~\ref{CLM:HATV-PREC-V}~and~\ref{CLM:C1-PREC-C2}}
\label{apx:hatv-preceq-v}

In this section we prove Lemma~\ref{CLM:HATV-PREC-V}~and~\ref{CLM:C1-PREC-C2}.

\begin{lemma}
\label{CLM:C1-PREC-C2}
$V_{|(\normone{v} \le \capsum_1)} \preceq V_{|(\normone{v} \le \capsum_2)}$ for any $c_1 \leq c_2$.
\end{lemma}


\begin{proof}
We are going to modify any $\tilde v$ dimension by dimension to reach $\hat v$. Define $n$ random functions $\sigma_1, \sigma_2, \ldots, \sigma_n$, where $\sigma_i(u)$ and $u$ only differs in $u_i$. In other words, $\sigma_i$ only modifies the $i$-th dimension of its input. Also define $\tau_i = \sigma_i \circ \sigma_{i - 1} \circ \cdots \circ \sigma_1$. Note that $\tau_n(v) \leq v$ as long as $\sigma_i(\tau_{i - 1}(v)) \leq \tau_{i - 1}(v)$ for all $i$.

In the first step, select a $\sigma_1$ such that
\begin{align*}
\Pr_{v \sim \tilde V}[(\sigma_1(v))_1] = \Pr_{v \sim \tilde V}\left[v_1 \left| \sum_{i = 1}^n v_i \leq \capsum_1 \right.\right].
\end{align*}
Then we separately deal with $v_2$'s according to their $(\tau_1(v))_1$ value. Select a $\sigma_2$ such that
\begin{align*}
\Pr_{v \sim \tilde V}[(\sigma_2(v))_2 \mid (\tau_1(v))_1] = \Pr_{v \sim \tilde V}\left[v_2 \left|(\tau_1(v))_1,\, (\tau_1(v))_1 + \sum_{i = 2}^n v_i \leq \capsum_1 \right.\right].
\end{align*}
We continue this procedure till we get $\sigma_n$. The $k$-th step will be setting
\begin{align*}
&\Pr_{v \sim \tilde V}[(\sigma_k(v))_k | (\tau_1(v))_1, \ldots, (\tau_{k - 1}(v))_{k - 1}]\\
= &\Pr_{v \sim \tilde V}\left[v_k \left|(\tau_1(v))_1, \ldots, (\tau_{k - 1}(v))_{k - 1},\, \sum_{i = 1}^{k - 1} (\tau_i(v))_i+ \sum_{i = k}^n v_i \leq \capsum_1 \right.\right].
\end{align*}

Next we show for all $k$, there exists $\sigma_k$ satisfying $(\sigma_k(x))_k \leq x_k$ for all $x$. This is simply first-order stochastic dominance for one-dimensional distributions, and it is equivalent to the following condition:
\begin{align*}
&\Pr_{v \sim \tilde V}[v_k \leq a \mid (\tau_1(v))_1, \ldots, (\tau_{k - 1}(v))_{k - 1}]\\
\leq &\Pr_{v \sim \tilde V}\left[v_k \leq a \left|(\tau_1(v))_1, \ldots, (\tau_{k - 1}(v))_{k - 1},\, \sum_{i = 1}^{k - 1} (\tau_i(v))_i+ \sum_{i = k}^n v_i \leq \capsum_1 \right.\right], \; \forall a \in \R.
\end{align*}

Writing $p_k = (v_{k + 1}, \ldots, v_n)$, $q_k = ((\tau_1(v))_1, \ldots, (\tau_{k - 1}(v))_{k - 1})$, and $r_k = (v_1, \ldots, v_{k - 1})$, the inequality above is true because
\begin{align*}
&\Pr_{v \sim \tilde V}\left[v_k \leq a \mid q_k,\, v_k + \normone{p_k} + \normone{q_k} \leq \capsum_1 \right]\\
= &\sum_{p_k, r_k} \Pr_{v \sim V}[p_k, r_k] \cdot 
\Pr_{v \sim V}[v_k \leq a \mid q_k, p_k, r_k,\, v_k  \leq \min(\capsum_1  - \normone{p_k} - \normone{q_k},\\ &\capsum_2 - \normone{p_k} - \normone{r_k})]\\
\geq &\sum_{p_k, r_k}\Pr_{v \sim V}[p_k, r_k] \cdot 
\Pr_{v \sim V}[v_k \leq a \mid q_k, p_k, r_k,\, v_k  \leq \min(\capsum_2  - \normone{p_k} - \normone{q_k},\\ &\capsum_2 - \normone{p_k} - \normone{r_k})]\\
= &\Pr_{v \sim \tilde V}\left[v_k \leq a \mid q_k\right].
\end{align*}

Therefore $\tau_n$ is the random function that defines coordinate-wise stochastic dominance, as every $\sigma_i$ satisfies $\sigma_i(x) \leq x$ for all $x$.
\end{proof}

\begin{lemma}
\label{CLM:HATV-PREC-V}
$V_{|(\normone{v} \le \capsum)} \preceq V$ for any $c > 0$.
\end{lemma}

\begin{proof}
It is implied by Lemma~\ref{CLM:C1-PREC-C2} when $c_1 = c$ and $c_2 = \max_{v \in \supp(V)} \normone{v}$.
\end{proof}

\end{document}